\documentclass[12pt,a4paper]{article}

\usepackage{amsthm}
\usepackage{amssymb}
\usepackage{amsmath}
\usepackage{graphicx}
\usepackage{pdfpages}
\usepackage{url}

\newcommand\se{\mathrm{se}}
\newcommand\SNR{\mathrm{SNR}}
\newcommand\E{\mathbb{E}}
\newcommand\R{\mathbb{R}}

\newtheorem{theorem}{Theorem}
\newtheorem{proposition}{Proposition}
\newtheorem{lemma}{Lemma}
\newtheorem{corollary}{Corollary}

\begin{document}

%% Here are the title, author names and addresses
\title{The Significance Filter, the Winner's Curse \\ and the Need to Shrink}

\author{E.W. van Zwet \footnote{Department of Biomedical Data Sciences, Leiden University Medical Center, Leiden, The Netherlands.} \qquad E.A. Cator \footnote{Faculty of Science, Radboud University, Nijmegen, The Netherlands.}}

\maketitle

\begin{abstract}
The ``significance filter'' refers to focusing exclusively on statistically significant results. Since frequentist properties such as unbiasedness and coverage are valid only before the data have been observed, there are no guarantees if we condition on significance. In fact, the significance filter leads to overestimation of the magnitude of the parameter, which has been called the ``winner's curse''. It can also lead to undercoverage of the confidence interval. Moreover, these problems become more severe if the power is low. While these issues clearly deserve our attention, they have been studied only informally and mathematical results are lacking. Here we study them from the frequentist and the Bayesian perspective. We prove that the relative bias of the magnitude is a decreasing function of the power and that the usual confidence interval undercovers when the power is less than 50\%. We conclude that failure to apply the appropriate amount of shrinkage can lead to misleading inferences.  
\end{abstract}

\section{Introduction}
The long-standing debate about the role of statistical significance in research \cite{rozeboom1960fallacy}, \cite{meehl1978theoretical} has recently intensified  \cite{wasserstein2016asa},\cite{benjamin2018redefine},\cite{wasserstein2019moving},\cite{mcshane2019abandon},\cite{amrhein2018remove} and \cite{ioannidis2019importance}.  Looking back to the beginning, we find that Ronald Fisher wrote in 1926 \cite{fisher1992arrangement}: 

\begin{quote}
``Personally, the writer prefers to set a low standard of significance at the 5 per cent point, and ignore entirely all results which fail to reach this level.'' 
\end{quote}

In other words, Fisher considered the familiar 5\% level to be quite liberal and recommended that results that fail to reach even that level can be safely ignored.  Now, more than 90 years later, Fisher's advice to apply the ``significance filter'' is widely followed. Recently, Barnett and Wren \cite{barnett2019examination} collected over 968,000 confidence intervals extracted from abstracts and over 350,000 intervals extracted from the full-text of papers published in Medline (PubMed) from 1976 to 2019. We converted these to $z$-values and their distribution is shown in Figure \ref{fig:z}. The under-representation of $z$-values between -2 and 2 is striking.

\begin{figure}[htp] \centering{
\includegraphics[scale=0.8]{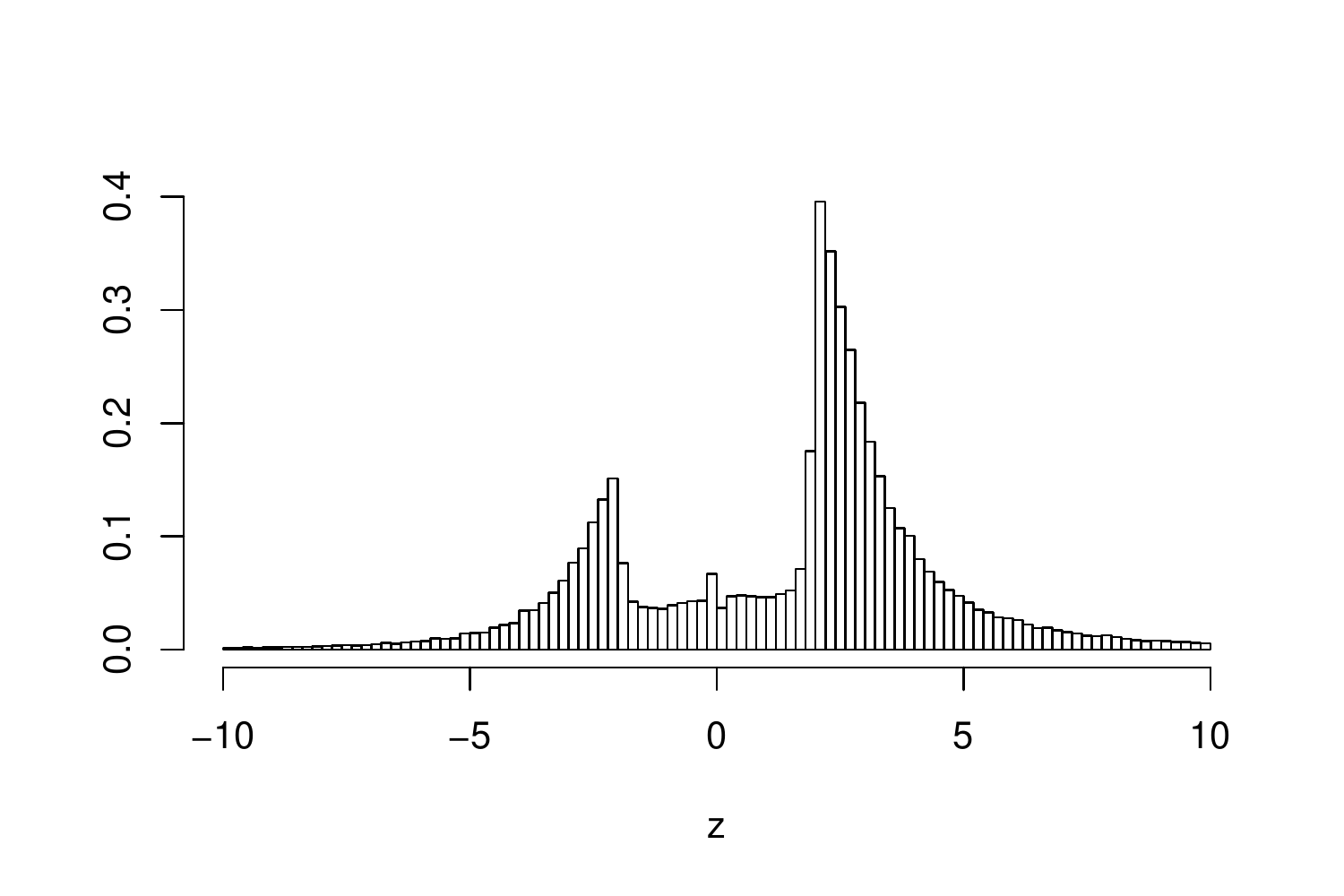}
\caption{The distribution of more than one million $z$-values from Medline (1976--2019).}\label{fig:z}
}
\end{figure}

As time and resources are always limited, it certainly makes sense to focus on significant results to avoid chasing noise. However, there is a problematic side-effect; considering only results that have reached statistical significance leads to overestimation \cite{ioannidis2008most}. This is sometimes called the ``winner's curse''. Moreover, it has been  demonstrated informally, i.e.\ by simulation, that the winner's curse is especially severe when the power is low \cite{ioannidis2008most},\cite{gelman2014beyond}. Here, we provide the first formal proof of this important fact.

As it turns out, low power is very common in the biomedical sciences \cite{button2013power},\cite{dumas2017low}. In particular, so-called pilot studies often have extremely low power. When such a study yields a significant result, the effect is likely grossly overestimated. Unfortunately, effect estimates from significant pilot studies are often used to inform the sample size calculation of a larger trial \cite{leon2011role},\cite{gelman2014beyond}. Low power also occurs when some correction is used to adjust for multiple comparisons. Such corrections are especially severe in genomics research, and the resulting overestimation of effects is well known \cite{goring2001large}. The recent suggestion to lower the significance level to improve reproducibility \cite{benjamin2018redefine} also reduces power, and therefore may backfire by aggravating the winner's curse \cite{mcshane2019abandon}.

While the winner's curse is a relatively well known phenomenon, mathematical results are lacking. In the first part of this paper, we study the winner's curse both the frequentist point of view. We prove that the relative bias in the magnitude is a decreasing function of the power. We also examine the effect of the significance filter on the coverage of confidence intervals and find it results in undercoverage when the power is less than 50\%.

In the second part of the paper, we study the significance filter from the Bayesian perspective. We conclude that it is necessary to apply shrinkage. We end the paper with a short discussion.

\section{The frequentist perspective}
Suppose that $b$ is a normally distributed,  unbiased estimator of $\beta$ with standard error $\se>0$. We have in mind that $\beta$ is some regression coefficient such as a difference of means, a slope, a log odds ratio or log hazard ratio, and we shall sometimes refer to $\beta$ as the ``effect''. 

\subsection{Bias of the magnitude}
By Jensen's inequality, $|b|$ is positively biased for $|\beta|$. Indeed, given $\beta$, $|b|$ has the folded normal distribution with mean
\begin{equation}
\E( |b| \mid \se, \beta)= |\beta| + \sqrt{\frac{2}{\pi}} \se\ e^{-\beta^2/2\se^2}  - 2 |\beta| \Phi\left( - \frac{|\beta|} {\se} \right).
\end{equation}

\begin{proposition}\label{prop:bias}
The bias $\E( |b| \mid \se, \beta) - |\beta|$ is positive for all $\se$ and $\beta$.  Moreover, it is decreasing in $|\beta|$ and increasing in $\se$.
\end{proposition}

\noindent
The proposition asserts that in low powered studies (small effects and large standard errors), the magnitude of the effect tends to be overestimated. For fixed $\se$, the bias $\E( |b| \mid \se, \beta) - |\beta|$ is maximal at $\beta=0$ where it is equal to $\sqrt{2/\pi}\, \se \approx 0.8\, \se$. 

Importantly, the bias in the magnitude becomes even larger if we condition on $|b|$ exceeding some threshold. This ``significance filter'' happens when journals preferentially accept results that are statistically significant (i.e.\ $|b|>1.96 \se$) but {\em also} when authors or readers choose to focus on such promising results as per Fisher's advice. We have the following extension of Proposition 1.

\begin{theorem} 
The conditional bias $\E( |b| \mid \se, \beta, |b|/\se>c) - |\beta|$ is positive for all $\se$ and $\beta$. Moreover, it is decreasing in $|\beta|$ and increasing in $\se$ and $c$.
\end{theorem}

\noindent
We define the {\it relative} conditional bias as 
$$\frac{\E( |b| \mid \se, \beta, |b|/\se>c) - |\beta|}{|\beta|}$$ 
and the exaggeration ratio or type M error \cite{gelman2014beyond} as $\E( |b| \mid \se, \beta, |b|>c) /|\beta|$.

\begin{corollary} 
The relative conditional bias is positive and and the exaggeration factor is greater than 1. Both depend on $\beta$ and $\se$ only through the signal-to-noise ratio (SNR) $|\beta|/\se$. Both quantities are decreasing in the SNR and increasing in $c$.
\end{corollary}
We illustrate this result in Figure \ref{fig:m}. Now the power for two-sided testing of $H_0:\beta=0$ at
level 5\% is
$$P(|b|>1.96\, \se \mid \beta,\se) = \Phi(\SNR-1.96) + 1 - \Phi(\SNR+1.96),$$
which is a strictly increasing function of the SNR. Hence, the relative conditional bias and the exaggeration factor are decreasing functions of the power, as was already noted on the basis of simulation in  \cite{ioannidis2008most} and \cite{gelman2014beyond}. 

\begin{figure}[htp] \centering{
\includegraphics[scale=0.8]{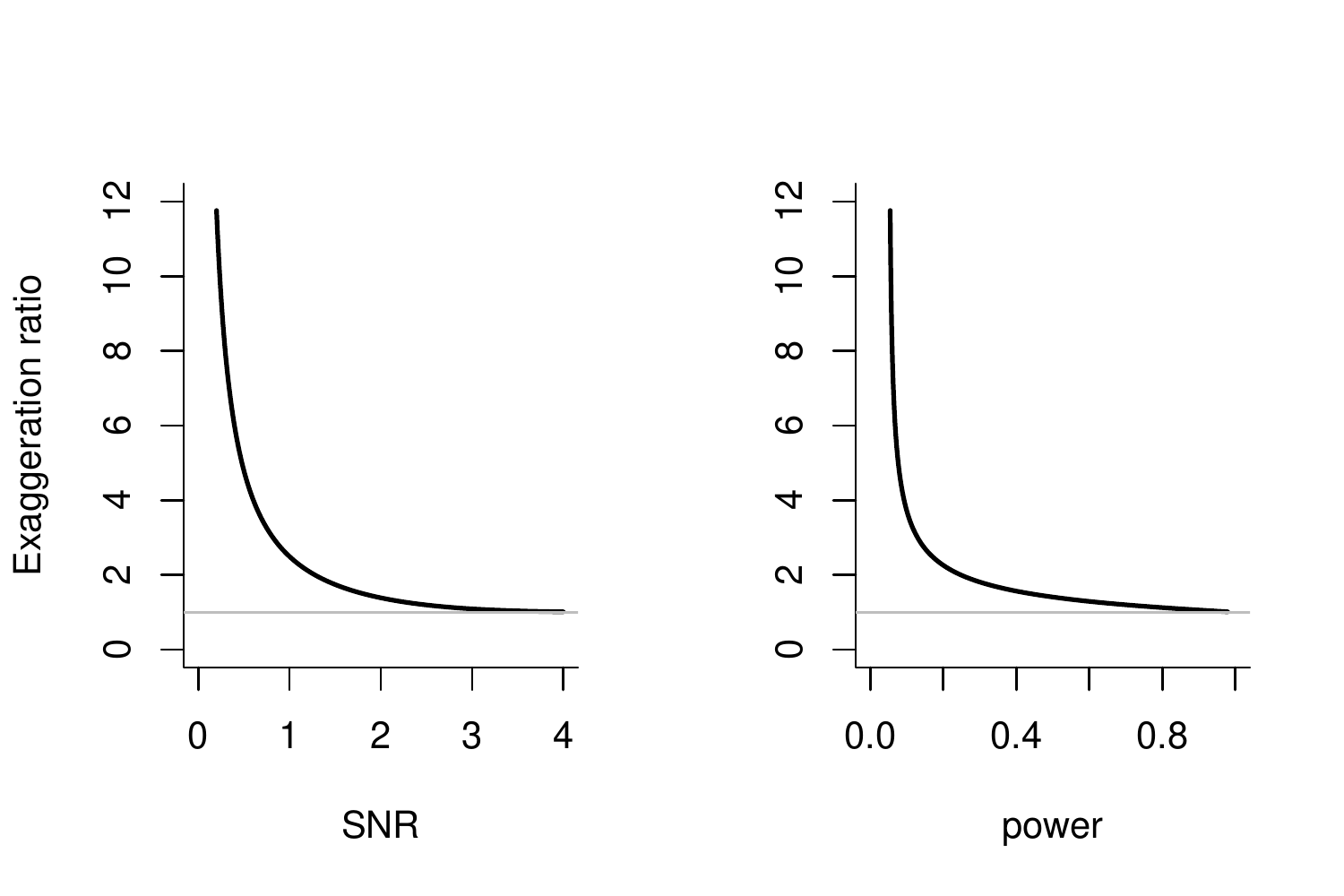}
\caption{The exaggeration factor as a function of the SNR and the power, when conditioning on significance at the 5\% level ($c=1.96$).}\label{fig:m}
}
\end{figure}

\subsection{Coverage}
The significance filter also has consequences for the coverage of confidence intervals. We start by recalling their definition. Suppose a random variable $X$ is distributed according to some distribution $f_\theta$. A $(1-\alpha) \times 100\%$ confidence set $S(X)$ is a random subset of the parameter space such that 
$$P(\theta \in S(X) \mid \theta)=1-\alpha,$$
for all $\theta$ \cite{lehmann2006testing}. A negatively biased semi-relevant (or recognizable) set $R$ is a subset of the sample space such that
$$P(\theta \in S(X) \mid \theta, X \in R)<1-\alpha,$$
for all $\theta$ . It is quite problematic if such a set $R$ exists, for is it still reasonable to report $S(X)$ with $(1-\alpha) \times 100\%$ confidence, after the event $X \in R$ has been observed?

Semi-relevant sets have been constructed in various situations, most notably in case of the standard one-sample $t$-interval \cite{lehmann2006testing}. Lehmann \cite{lehmann2006testing} called the existence of certain relevant sets ``an embarrassment to confidence theory''. Now suppose $b$ is normally distributed with mean $\beta$ and known standard deviation $\se$. If we define
$$S(b)=\{ \beta : |b-\beta|/\se < z_{1-\alpha/2}  \}$$
where $z_{1-\alpha/2}$ is the $1-\alpha/2$ quantile of the standard normal distribution, then we have the following confidence statement
$$P(\beta \in S(b) \mid \beta,\se) = P(|b-\beta|/\se < z_{1-\alpha/2}  \mid \beta,\se)=1-\alpha,$$
for all $0 < \alpha < 1$, $\beta$ and $\se>0$.  Lehmann \cite{lehmann2006testing} shows that in this particular setting, there do not exist any negatively biased semi-relevant  sets. This is certainly reassuring. However, we if $c>0$, then the conditional coverage
$$P(|b-\beta|/ \se < z_{1-\alpha/2} \mid \beta, \se, |b|/\se>c)$$
depends on $\beta$ and $\se$. This dependence is not simple. For instance, it is {\em not} monotone in $\beta$. We do have the following Theorem.

\begin{theorem}
Suppose $b$ is normally distributed with mean $\beta$ and standard deviation $\se$. If the SNR  $|\beta|/\se$ is less than $z=z_{1-\alpha/2}$ then
\begin{equation}
P(|b-\beta| / \se< z  \mid \beta, \se, |b|/\se>z) < P(|b-\beta| / \se< z \mid \beta, \se) =1-\alpha.
\end{equation}
\end{theorem}

Note that if the SNR is equal to $z_{1-\alpha/2}$, then the power for testing $H_0 : \beta = 0$ at level $\alpha$ is slightly more than 50\%. So the Theorem implies that if we have a significant result while the power is 50\% or less, then the confidence interval will not reach its nominal coverage.

The result is quite sharp. By inspecting the proof, we can see that if the SNR is slightly larger than $z_{1-\alpha/2}$ then the conditional coverage {\em exceeds} the nominal (unconditional) coverage.

\section{The Bayesian Perspective}
Bayesian inference is valid conditionally on the data, and so the significance filter should not pose any difficulties. On the other hand, Bayesian estimators are naturally biased. In this section we compare the performance of the unbiased estimator $b$ and the Bayes estimator. 

Let us assume that $\beta$ has a normal prior distribution with mean 0 and known standard deviation $\tau >0$. The conditional distribution of $\beta$ given $b$ is normal with mean $b^*=\tau^2 b/(\se^2 + \tau^2)$ and variance $v=\se^2\tau^2/(\se^2 + \tau^2)$. We will write $s=\sqrt{v}$. Note that $b^*$ is the Bayes estimator (under squared error loss) of $\beta$. Clearly, $|b^*|<|b|$ and for that reason $b^*$ is called a shrinkage estimator.

We can evaluate $b$ and $b^*$ as estimators of $\beta$ conditionally on the parameter and averaged over the distribution of the  data, which is the frequentist point of view. Alternatively, we can condition on the data and average over the distribution of the parameter, which is the Bayesian point of view. We have the following nicely symmetric situation, where we consider $\se$ and $\tau$ to be fixed and known.

\begin{align}
\E(b - \beta \mid \beta) &= 0 &\text{and}\quad \E(b^* - \beta \mid \beta) &= -\frac{\se^2}{\se^2 + \tau^2} \beta \\
\E(b - \beta \mid b) &= \frac{\se^2}{\se^2 + \tau^2}b &\text{and} \quad \E(b^* - \beta \mid b) &= 0 
\end{align}
So, from the frequentist point of view, $b$ is unbiased for $\beta$ and $b^*$ is biased. However, from the Bayesian point of view, it is the other way around!  

\subsection{Bias of the magnitude}
Now, if we are interested in the magnitude of $\beta$, then we could take the posterior mean of $|\beta|$ as an estimator. However, it is still relevant to evaluate the performance of $|b^*|$ as an estimator of $|\beta|$ from the Bayesian point of view. Conditionally on $s$ and $b^*$, $\beta$ has the normal distribution with mean $b^*$ and standard deviation $s$ and hence $|\beta|$ has the folded normal distribution. Similarly to Proposition 1, we have the following.
%\begin{equation}
%\E( |\beta| \mid s,b^*)= |b^*| + \sqrt{\frac{2}{\pi}} s e^{-{b^*}^2/2 s^2}  - 2 |b^*| \Phi\left( - \frac{|b^*|} {s} \right).
%\end{equation}

\begin{proposition} 
The difference  $\E( |\beta| \mid s, b^*) - |b^*|$ is positive. It is decreasing in $|b^*|$ and increasing in $s$. Moreover, the difference vanishes as $|b^*|$ tends to infinity.
\end{proposition}

\noindent
So, conditionally on the data, $|b^*|$ underestimates $|\beta|$ on average, but the difference disappears if we focus on large or significant effects. So now the significance filter actually {\em reduces} the bias in the magnitude! In other words, shrinkage lifts the winner's curse.

\bigskip \noindent
So far, we have conditioned either on the parameter or the data, and averaged over the other. However, in practice we do not keep the parameter fixed and repeat the experiment many times. We also do not keep the data fixed and vary the parameter. So, it is also relevant to consider the performance of $b$ and $b^*$ on average over the distribution of {\em both} the parameter {\em and} the data. If the distribution of the parameter represents some field of research, then this averaging will provide insight into how our statistical procedures perform when used repeatedly in that field.

Under our simple model, the marginal distribution of $b$ is normal with mean zero and variance $\se^2+\tau^2$ and the marginal distribution of $b^*$ is normal with mean zero and variance $\tau^4/(\se^2+\tau^2)$.  So, trivially, $\E(b)=\E(b^*)=\E(\beta)=0$. Moreover, it is easy to see that the variance of $b^*$ is less than the variance of $b$. Marginally, $|\beta|$, $|b|$ and $|b^*|$ have half-normal distributions with means
\begin{equation}
\E \left| b^* \right| =\frac{\tau^2}{\sqrt{\se^2 + \tau^2}}  \sqrt{\frac{2}{\pi}}, \quad \E|\beta| = \tau \sqrt{\frac{2}{\pi}},  \quad \E|b| = \sqrt{\se^2 + \tau^2}  \sqrt{\frac{2}{\pi}}
\end{equation}
It is easy to see that
\begin{equation}
\E \left| b^* \right|  < \E |\beta| < \E|b|.
\end{equation}
Negative bias is more conservative than positive bias, and that may be preferable in many situations. It is interesting to note that the factor by which $|b|$ {\em over}estimates $|\beta|$ is the same as the factor by which $|b^*|$ {\em under}estimates it. That is,
\begin{equation}
\frac{\E|b|}{ \E|\beta|} = \frac{ \E|\beta|}{\E|b^*|} = \frac{\sqrt{\se^2 + \tau^2}}{\tau} .
\end{equation}
Moreover, the following proposition says that the bias of $|b^*|$ is smaller (on average) than the bias of $|b|$.
\begin{proposition}
Suppose $\beta$ has a normal prior distribution with mean 0 and standard deviation $\tau >0$. Suppose that conditionally on $\beta$, $b$ is normally distributed with mean $\beta$ standard error $\se>0$. Let $b^*=\E(\beta \mid b)$, then
\begin{equation}
\E( |b| - |\beta|) > \E(|\beta| -  |b^*|).
\end{equation}
\end{proposition}

\noindent
Most importantly, however, while the bias of $|b|$ increases as we condition on $|b|$ exceeding some threshold, the bias of $|b^*|$ vanishes!
\begin{theorem} 
As $c$ goes to infinity, $\E(|b^*| -|\beta| \mid |b|>c)$ vanishes.
\end{theorem}

\subsection{Coverage}
We now return to the coverage issue we discussed in section 2.2. It might seem that Theorem 2 is not much of a problem in practice because conditional on a significant result, the power is unlikely to be small. But such an argument would depend on the  (prior) distribution of the signal-to-noise ratio $|\beta|/\se$. We have the following result.

\begin{theorem}
Suppose $\beta$ and $\se$ are distributed such that the SNR $|\beta|/\se$ has a decreasing density and $|\beta|/\se$ and $\se$ are independent. Also suppose that conditionally on $\beta$ and $\se$, $b$ is normally distributed with mean $\beta$ and standard deviation $\se$. For every $0<\alpha < 1$
\begin{equation}
P(|b-\beta| < z_{1-\alpha/2}\se \mid  |b|/\se>c) < 1-\alpha.
\end{equation}
\end{theorem}

This result suggest that across research fields where of the SNR has a decreasing density, confidence interval undercover on average. But how realistic is it to assume such a decreasing density? Clearly, it would imply a decreasing density of the absolute $z$-value, and this is certainly not the case in  Figure \ref{fig:z}. However we believe that this is due to selective reporting. 

We have made an effort to collect an unselected sample of $z$-values as follows. It is a fairly common practice in the life sciences to build multivariate regression models by ``univariable screening''. First, the researchers run a number of univariable regressions for all predictors that they believe could have an important effect. Next, those predictors with a $p$-value below some threshold are selected for the multivariate model. While this approach is statistically unsound, we believe that the univariable regressions should be largely unaffected by selection on significance, simply because that selection is still to be done. For further details, we refer to \cite{van2019default}. We do note that in that article, we discarded $p$-values below 0.001, but these are included here. 

We have collected 732 absolute $z$-values from 51 recent articles from Medline. We show the distribution in Figure \ref{fig:z2} which suggest a decreasing distribution of the absolute $z$-values, which implies a that the distribution of the SNR is decreasing as well. 

\begin{figure}[htp] \centering{
\includegraphics[scale=0.8]{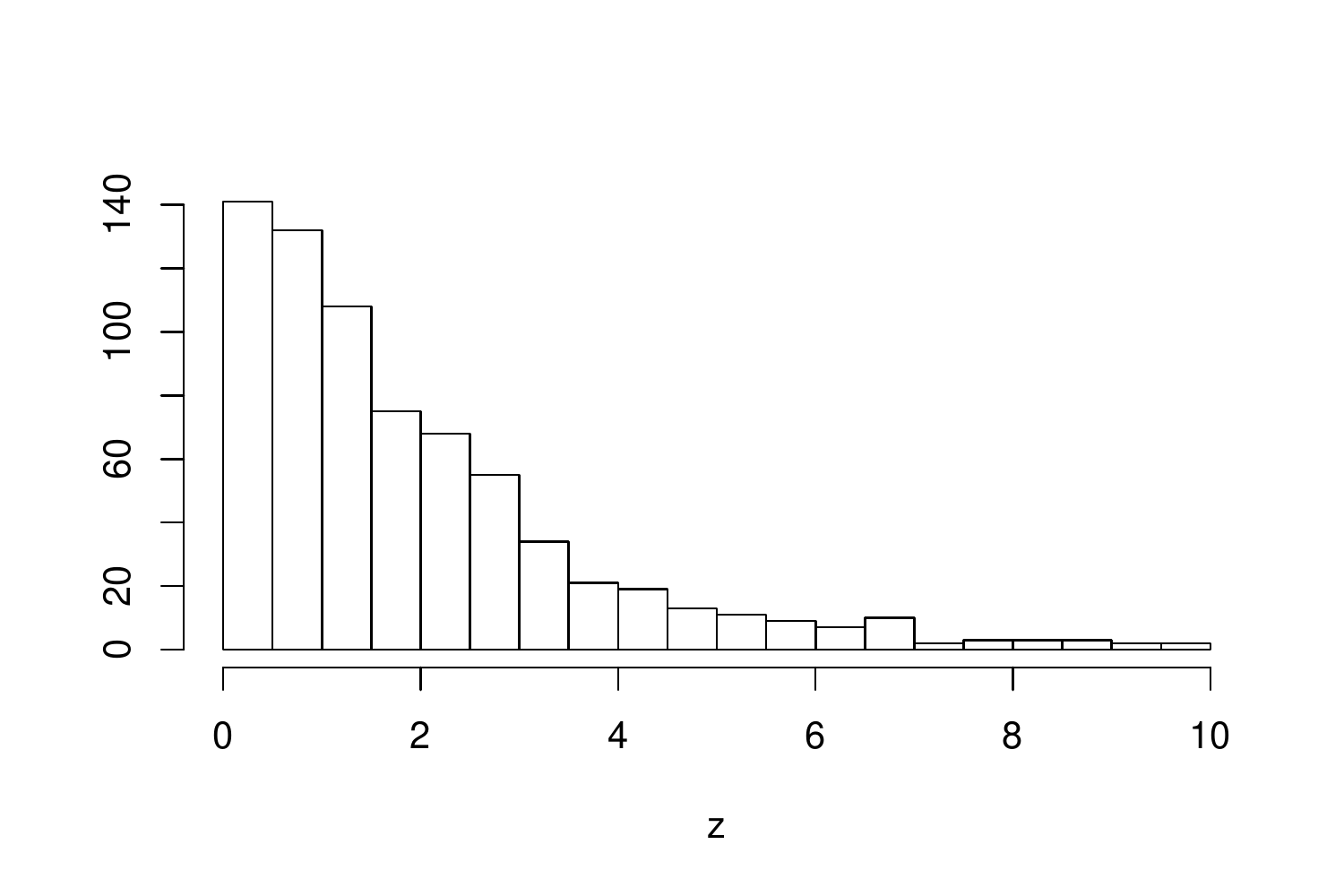}
\caption{The distribution of 732 absolute $z$-values from Medline. Here we made an effort to avoid the significance filter.}\label{fig:z2}
}
\end{figure}

\section{Discussion}
In this paper we have considered the generic situation where we have an unbiased, normally distributed estimator $b$ of a parameter $\beta$, with known standard error $\se$. Frequentist properties, such as the unbiasedness of $b$ and the coverage of the confidence interval are only meaningful before the data have been observed. Once the data are in, they become meaningless since $b$ is just some fixed number and the confidence interval either covers $\beta$ or it does not. Nothing more can be said without specifying a (prior) distribution for $\beta$.

However, suppose we condition not on $(b,\se)$ but only on the event $|b|>1.96\, \se$. That is, we condition on statistical significance at the 5\% level. Now $b$ is still random and we can talk about bias and coverage.  Conditionally on significance, $b$ is biased away from zero. This tendency to overestimate the magnitude of significant effects is sometimes called the ``winner's curse''. It is especially severe when the signal-to-noise ratio $|\beta|/\se$ is low. Also, if the SNR is low, then conditionally on significance the confidence interval will undercover. By providing mathematical proofs of these facts, we hope to contribute to the awareness of these very serious problems.

The goal of hypothesis testing is to try to avoid chasing noise, which is perfectly reasonable. However, the consequence of focusing on significant results is that all the nice frequentist properties no longer hold. Many proposals have been made to address this issue.  From a frequentist point of view, one could condition throughout on statistical significance. See, for example, \cite{ghosh2008estimating} and references therein. Alternatively, one can take a Bayesian approach, such as proposed by \cite{xu2011bayesian} and ourselves \cite{van2019default}. Of course, the Bayesian approach relies on correct specification of the prior. 

Shrinkage is often viewed as a method to achieve a lower mean squared error by reducing the variance at the expense of  increasing the bias. Our most important point is that it is necessary to apply shrinkage to {\em reduce} the bias that results from focusing on interesting results.

\appendix
\setcounter{proposition}{0}
\setcounter{corollary}{0}
\setcounter{theorem}{0}

\section{Appendix}
\begin{proposition}
The bias $\E( |b| \mid \se, \beta) - |\beta|$ is positive for all $\se$ and $\beta$.  Moreover, it is decreasing in $|\beta|$ and increasing in $\se$.
\end{proposition}
\begin{proof}
This is a special case of Theorem 1.
\end{proof}

\bigskip

\begin{theorem} 
The conditional bias $\E( |b| \mid \se, \beta, |b|/\se>c) - |\beta|$ is positive for all $\se$ and $\beta$. Moreover, it is decreasing in $|\beta|$ and increasing in $\se$ and $c$.
\end{theorem}

\begin{proof}
Let $Z$ be a standard normal random variable and define 
\begin{equation}\label{def:g}
g(\theta,c) = \E\left(|\theta+Z|-\theta \mid |\theta+ Z|\geq c\right).
\end{equation}
Since
\begin{equation}
\E( |b| \mid \se, \beta, |b|/\se>c) - |\beta| = \se g(|\beta|/\se,c)
\end{equation}
it is clear that it is enough to prove that $g(\theta,c)$ is decreasing in $\theta>0$ and increasing in $c>0$. 

Suppose $c_1<c_2$. For any random variable $X$ we have that
$$ \E(X\mid X\geq c_2) \geq \E(X),$$
since $\E(X\mid X<c_2)\leq c_2\leq \E(X\mid X\geq c_2)$ and $\E(X)$ is a convex combination of the two conditional expectations. Now we can replace $X$ by $\tilde{X}=\E(X\mid X\geq c_1)$, and we conclude that $g$ is increasing in $c>0$. 

To prove that $g(\theta,c)$ is decreasing in $\theta>0$, note that the density of $|Z + \theta|$ is given by
$$ f(y) = 
\begin{cases}
\phi(y-\theta) + \phi(y+\theta),\quad &y \geq 0\\
0, &y<0.
\end{cases}
$$
Here $\phi$ is the standard normal density. Using that $z\phi(z) = -\phi'(z)$,
\begin{align}\label{eq:condexp}
g(\theta,c) & = \frac{\int_c^\infty (y-\theta)( \phi(y-\theta) + \phi(y+\theta))\,dy}{P(|Z + \theta|\geq c)}\nonumber\\
& = \frac{\int_c^\infty -\phi' (y-\theta) -\phi' (y+\theta) - 2\theta\phi(y+\theta)\,dy}{P(|Z + \theta|\geq c)}\nonumber\\
& = \frac{\phi(c-\theta) + \phi(c+\theta) - 2\theta(1-\Phi(c+\theta))}{2-\Phi(c-\theta)-\Phi(c+\theta)}.
\end{align}
We split $g(\theta,c)$ into the numerator and the denominator:
\[ N = \phi(c-\theta) + \phi(c+\theta) - 2\theta\left(1-\Phi(c+\theta)\right)\]
and
\[ D = 2-\Phi(c-\theta)-\Phi(c+\theta).\]
Now it is enough to check that
\[ D\cdot\frac{\partial N}{\partial \theta} \leq N\cdot\frac{\partial D}{\partial \theta}.\]
So
\begin{align*}
\frac{\partial N}{\partial \theta}  &= -\phi'(c-\theta) +\phi'(c+\theta) -2( 1-\Phi(c+\theta)) +2\theta\phi(c+\theta)\\
& = (c-\theta) (\phi(c-\theta) - \phi(c+\theta)) -2( 1-\Phi(c+\theta)),
\end{align*}
and
\[ \frac{\partial D}{\partial \theta} = \phi(c-\theta) - \phi(c +\theta).\]
Introduce
\[ z_- = c-\theta \mbox{ and } z_+ = c+\theta.\]
Then
\begin{align} \label{eq:ineq}
D\cdot\frac{\partial N}{\partial \theta} &\leq N\cdot\frac{\partial D}{\partial \theta} &\iff \nonumber\\
D\cdot z_-\left(\phi\left(z_-\right) - \phi\left(z_+\right) \right)& -2D\cdot\left(1-\Phi\left(z_+\right)\right)   \leq N \cdot \left(\phi\left(z_-\right) - \phi\left(z_+\right) \right) & \iff \nonumber\\
 \left(\phi\left(z_-\right) - \phi\left(z_+\right)\right)& \left( D\cdot z_- - N\right) \leq 2D\cdot\left(1-\Phi\left(z_+\right)\right).
\end{align}
The right-hand-side of \eqref{eq:ineq} is clearly positive, and it is not hard to see that the first factor of the left-hand-side is also positive: for $|z_-|\leq z_+$ we have
\[ \phi\left(z_-\right) - \phi\left(z_+\right) \geq 0.\]
Therefore, we can show that \eqref{eq:ineq} is true, if we can show that
\begin{equation}\label{eq:neg} 
D\cdot z_- - N \leq 0.
\end{equation}
We can see that
\begin{align}\label{eq:miracle}
D\cdot z_- - N &=  z_-\left(1-\Phi(z_-)\right) + z_-\left(1-\Phi(z_+)\right) - \phi(z_-) - \phi(z_+) + 2\theta (1-\Phi(z_+))\nonumber \\
&= z_-\left(1-\Phi(z_-)\right) - \phi(z_-) + z_+\left(1-\Phi(z_+)\right) - \phi(z_+).
\end{align}
We now use the fact that for all $z\in \R$, 
\[ z(1-\Phi(z))-\phi(z)\leq 0,\]
which follows from the fact that the derivative of this function (i.e. $1-\Phi(z)$) is positive, and the limit for $z\to\infty$ equals $0$. So \eqref{eq:miracle} is indeed negative, which proves \eqref{eq:ineq}, and therefore the fact that $g(\theta,c)$ is decreasing in $\theta \geq 0$.

\end{proof}

\bigskip

\begin{corollary} 
The relative conditional bias is positive and and the exaggeration factor is greater than 1. Both depend on depend on $\beta$ and $\se$ only through the signal-to-noise ratio (SNR) $|\beta|/\se$. Both are decreasing in $|\beta|/\se$ and increasing in $c$.
\end{corollary}
\begin{proof}
Recall the definition of the function $g(\theta,c)$ from (\ref{def:g}). The relative bias is equal to $\se g(|\beta|/\se,c)/|\beta|$ and the exaggeration factor is $\se g(|\beta|/\se,c)/|\beta| + 1$. We refer to the proof of Theorem 1 where we show that $g(\theta,c)$ is decreasing in $\theta \geq 0$ and increasing in $c\geq 0$. This also establishes the present claim.
\end{proof}

\bigskip

\begin{proposition} 
The difference  $\E( |\beta| \mid s, b^*) - |b^*|$ is positive. It is decreasing in $|b^*|$ and increasing in $s$. Moreover, the difference vanishes as $|b^*|$ tends to infinity.
\end{proposition}
\begin{proof}
Comparing to Proposition 1, we see that this is also a special case of Theorem 1.
\end{proof}

\bigskip

\begin{theorem}
Suppose $b$ is normally distributed with mean $\beta$ and standard deviation $\se$. If $|\beta|/\se \leq z=z_{1-\alpha/2}$ then
\begin{equation}
P(|b-\beta| / \se< z  \mid \beta, \se, |b|/\se>z) < P(|b-\beta| / \se< z \mid \beta, \se) =1-\alpha.
\end{equation}
\end{theorem}
\begin{proof}
There is no loss of generality if we assume $\beta>0$ and $\se=1$. In this proof, we will drop conditioning on $\beta$ and $\se$ from our notation. In fact, without loss of generality we will prove the corresponding statement for $X\sim N(\mu,1)$.  Also, it is more convenient to work with the complementary event $|X-\mu| > z$. Since
$$ P(|X-\mu| > z  \mid |X|>z) = \frac{P(|X|>z| \mid |X-\mu|> z)P(|X-\mu| > z)}{P(|X|>z)}$$
is suffices to prove that
$$P(|X|>z| \mid |X-\mu|> z) - P(|X|>z) > 0$$
for all $0<\mu \leq z$. Now,
\begin{align*}
P(|X|>z| &\mid |X-\mu|> z) - P(|X|>z) =\\
&= P(|X|>z| \mid X-\mu > z)/2 +P(|X|>z \mid X-\mu< -z)/2 \\
&\ \ \ - P(|X|>z) \\
&= \frac12 + P(X>z \mid X-\mu< -z)/2 +P(X< -z \mid X-\mu< -z)/2 \\
&\ \ \ - P(X>z) - P(X<-z) \\
&=\frac12 + \frac{\Phi(-z-\mu)}{2\Phi(-z)} - 1 + \Phi(z-\mu) - \Phi(-z-\mu).
\end{align*}
Taking the derivative with respect to $\mu$, it is easy to see that this expression is decreasing in $\mu>0$. Moreover, if we take $\mu=z$, then we get $\Phi(-2z)/2\Phi(-z) - \Phi(-2z)$, which is positive because $\Phi(-z)<1/2$.
\end{proof}

\bigskip

\begin{proposition}
Suppose $\beta$ has a normal prior distribution with mean 0 and standard deviation $\tau >0$. Suppose that conditionally on $\beta$, $b$ is a normally distributed with mean $\beta$ standard error $\se>0$. Let $b^*=\E(\beta \mid b)$, then
\begin{equation}
\E( |b| - |\beta|) > \E(|\beta| -  |b^*|).
\end{equation}
\end{proposition}
\begin{proof}
We have to show that
$$\sqrt{\se^2 + \tau^2}  - \tau \geq \tau - \frac{\tau^2}{\sqrt{\se^2 + \tau^2}}.$$
Multiplying by $\sqrt{\se^2 + \tau^2}$ and rearranging we obtain
$$\se^2 + \tau^2 -2 \tau \sqrt{\se^2 + \tau^2} + \tau^2 \geq 0.$$
The left hand side of this equality is equal to $(\sqrt{\se^2 + \tau^2}  - \tau)^2$ which is clearly positive unless $\se$ is zero.
\end{proof}

 \bigskip
 
\begin{theorem} 
As $c$ goes to infinity, $\E(|b^*| -|\beta| \mid |b|>c)$ vanishes.
\end{theorem}
\begin{proof}
Since the marginal distribution of $b$ is symmetric around zero, we have for positive $c$
\begin{equation}\E(|b^*| \mid |b|>c) = \frac{\tau^2}{\se^2 + \tau^2} \E( |b| \mid |b| >c)=\frac{\tau^2}{\se^2 + \tau^2} \E( b \mid b >c)\end{equation}
Conditionally on $b>c$, $b$ has the truncated normal distribution. Hence
\begin{align}
\frac{\tau^2}{\se^2 + \tau^2} \E( b \mid b >c) & = \frac{\tau^2}{\se^2 + \tau^2} \sqrt{\se^2 + \tau^2} \frac{\varphi(c/\sqrt{\se^2 + \tau^2})}{1-\Phi(c/\sqrt{\se^2 + \tau^2})} \notag \\
& = \frac{\tau^2 \varphi(c/\sqrt{\se^2 + \tau^2})}{\sqrt{\se^2 + \tau^2}(1-\Phi(c/\sqrt{\se^2 + \tau^2}))} \label{bias}
\end{align}
Turning to $|\beta|$, we have by symmetry, 
\begin{equation}
\E(|\beta| \mid |b| > c) = \E(|\beta| \mid b > c).
\end{equation}
Moreover, 
\begin{equation}\E(|\beta| \mid b > c) = \E(\beta \mid b>c) + \E(|\beta|-\beta\mid b>c).\end{equation}
By a result due to Rosenbaum \cite{rosenbaum1961moments} concerning the mean of a truncated bivariate normal distribution, we have
\begin{equation}
\E(\beta \mid b>c) = \frac{\tau^2 \varphi(c/\sqrt{\se^2 + \tau^2})}{\sqrt{\se^2 + \tau^2}(1-\Phi(c/\sqrt{\se^2 + \tau^2}))}.
\end{equation}
Since this expression is equal to (\ref{bias}), we only need to show that
\[ \lim_{c\to\infty} \E(|\beta|-\beta\mid b>c) = 0.\]
Since $P(\beta<0\mid b>c)\to 0$ as $c\to \infty$, this is clearly true.
\end{proof}

\bigskip
To prove Theorem 4, we use the following Lemma.

\begin{lemma}
Let $X$ be a random variable and $g$ an increasing function that is not constant on the support of $X$. Then for every $x$ such that $P(X>x)>0$
$$\E(g(X) \mid X < x) < \E(g(X)).$$
\end{lemma}
\begin{proof}
$\E(g(X))$ is a convex combination of $\E(g(X) \mid X < x)$ and $\E(g(X) \mid X \geq x)$. Since $g$ is increasing
$$\E(g(X) \mid X < x) \leq \E(g(X) \mid X \geq x).$$ 
If $g$ is not constant on the support of $X$, then the inequality is strict and the claim follows.
\end{proof}

\bigskip

To prove Theorem 4, we first prove the following Proposition.
\begin{proposition}
Suppose $\mu$ is distributed such that $|\mu|$ has a decreasing density $f$. Also suppose $Z$ is independent of $\mu$ and has a distribution which is symmetric around zero and supported on the whole real line. Let $X=Z+\mu$. For every positive $c$ and $z$
\begin{equation}
P(|X-\mu| < z \mid |X|>c) < P(|X-\mu| < z).
\end{equation}
\end{proposition}
\begin{proof}
 Since $Z=X-\mu$, we can rewrite the claim as
 $$P(|Z| < z \mid |Z+\mu|>c) < P(|Z| < z).$$
 Because,
 $$P(|Z|<z \mid |Z+\mu|>c) = \frac{P( |Z+\mu| > c \mid |Z|<z) P(|Z|<z)}{P( |Z+\mu|>c)}$$
 the claim is equivalent to
 \begin{equation}
 P( |Z+\mu|>c \mid |Z|<z) < P( |Z+\mu|>c),
 \end{equation}
 for all positive $z$ and $c$. If we define
 $$g(z)=P( |Z+\mu|>c \mid |Z|=z)$$
 then
 $$P( |Z+\mu|>c \mid |Z|<z) = \E(g(|Z|)  \mid |Z|<z).$$ 
 Now we can use Lemma 1 to prove our claim by showing that $g$ is increasing and not constant.
 Since the distribution of $Z$ is symmetric around zero, its sign and magnitude are independent. Therefore,
 \begin{equation}
 g(z)= P( |Z+|\mu||>c \mid |Z|=z).
 \end{equation}
 Using the independence of $Z$ and $\mu$, we have
  \begin{align*}
  g(z) &=  P( |z+|\mu||>c)/2 + P( |-z+|\mu||>c)/2 \\
 &= P( z+|\mu| >c)/2 + P( z+|\mu| < -c)/2 \\
 &+ P( -z+|\mu| >c)/2 + P( -z+|\mu| < -c)/2 \\
 &= \begin{cases}
 P(|\mu| >c-z)/2 + P(|\mu| >c+z)/2,  &\text{if }z \leq c \\
 \frac12 + P(|\mu| >c+z)/2 + P(|\mu| < -c+z)/2, &\text{if }z > c 
 \end{cases}
 \end{align*}
 Taking the derivative, we have 
 \begin{align*}
  g'(z) &= 
 \begin{cases}
 f(c-z) - f(c+z),  &\text{if }z \leq c \\
 -f(c+z) + f(-c+z), &\text{if }z > c
 \end{cases}\\
 &=f(|z-c|) - f(z+c) \geq 0.
 \end{align*}
 $f$ is not constant since it is a (proper) density. It follows that $g'$ cannot be identically zero and hence $g$ is not constant either.
\end{proof}

\bigskip

\begin{theorem}
Suppose $\beta$ and $\se$ are distributed such that $|\beta|/\se$ has a decreasing density and $|\beta|/\se$ and $\se$ are independent. Also suppose that conditionally on $\beta$ and $\se$, $b$ is normally distributed with mean $\beta$ and standard deviation $\se$. For every $0<\alpha < 1$
\begin{equation}
P(|b-\beta| < z_{1-\alpha/2}\se \mid  |b|/\se>c) < 1-\alpha.
\end{equation}
\end{theorem}

\begin{proof}
For every $\se>0$, it follows from Proposition 1 that
$$P(|X-\mu| < z_{1-\alpha/2}\se \mid \se, |X|/\se>c) < P(|X-\mu| < z_{1-\alpha/2}\se \mid \se) = 1-\alpha.$$
Averaging over $\se$, the claim follows.
\end{proof}

%\printbibliography[title={\bf{References}}]
%\bibliographystyle{plain}
%\bibliography{literature}

\begin{thebibliography}{10}

\bibitem{amrhein2018remove}
Valentin Amrhein and Sander Greenland.
\newblock Remove, rather than redefine, statistical significance.
\newblock {\em Nature Human Behaviour}, 2(1):4, 2018.

\bibitem{barnett2019examination}
Adrian~Gerard Barnett and Jonathan~D Wren.
\newblock Examination of cis in health and medical journals from 1976 to 2019:
  an observational study.
\newblock {\em BMJ Open}, 9(11), 2019.

\bibitem{benjamin2018redefine}
Daniel~J Benjamin, James~O Berger, Magnus Johannesson, Brian~A Nosek, E-J
  Wagenmakers, Richard Berk, Kenneth~A Bollen, Bj{\"o}rn Brembs, Lawrence
  Brown, Colin Camerer, et~al.
\newblock Redefine statistical significance.
\newblock {\em Nature Human Behaviour}, 2(1):6, 2018.

\bibitem{button2013power}
Katherine~S Button, John~PA Ioannidis, Claire Mokrysz, Brian~A Nosek, Jonathan
  Flint, Emma~SJ Robinson, and Marcus~R Munaf{\`o}.
\newblock Power failure: why small sample size undermines the reliability of
  neuroscience.
\newblock {\em Nature Reviews Neuroscience}, 14(5):365, 2013.

\bibitem{dumas2017low}
Estelle Dumas-Mallet, Katherine~S Button, Thomas Boraud, Francois Gonon, and
  Marcus~R Munaf{\`o}.
\newblock Low statistical power in biomedical science: a review of three human
  research domains.
\newblock {\em Royal Society open science}, 4(2):160254, 2017.

\bibitem{fisher1992arrangement}
Ronald~A Fisher.
\newblock The arrangement of field experiments.
\newblock In {\em Breakthroughs in statistics}, pages 82--91. Springer, 1992.

\bibitem{gelman2014beyond}
Andrew Gelman and John Carlin.
\newblock Beyond power calculations: Assessing type s (sign) and type m
  (magnitude) errors.
\newblock {\em Perspectives on Psychological Science}, 9(6):641--651, 2014.

\bibitem{ghosh2008estimating}
Arpita Ghosh, Fei Zou, and Fred~A Wright.
\newblock Estimating odds ratios in genome scans: an approximate conditional
  likelihood approach.
\newblock {\em The American Journal of Human Genetics}, 82(5):1064--1074, 2008.

\bibitem{goring2001large}
Harald~HH G{\"o}ring, Joseph~D Terwilliger, and John Blangero.
\newblock Large upward bias in estimation of locus-specific effects from
  genomewide scans.
\newblock {\em The American Journal of Human Genetics}, 69(6):1357--1369, 2001.

\bibitem{ioannidis2008most}
John~PA Ioannidis.
\newblock Why most discovered true associations are inflated.
\newblock {\em Epidemiology}, 19(5):640--648, 2008.

\bibitem{ioannidis2019importance}
John~PA Ioannidis.
\newblock The importance of predefined rules and prespecified statistical
  analyses: Do not abandon significance.
\newblock {\em Jama}, 321(21):2067--2068, 2019.

\bibitem{lehmann2006testing}
Erich~L Lehmann and Joseph~P Romano.
\newblock {\em Testing statistical hypotheses}.
\newblock Springer Science \& Business Media, 2006.

\bibitem{leon2011role}
Andrew~C Leon, Lori~L Davis, and Helena~C Kraemer.
\newblock The role and interpretation of pilot studies in clinical research.
\newblock {\em Journal of psychiatric research}, 45(5):626--629, 2011.

\bibitem{mcshane2019abandon}
Blakeley~B McShane, David Gal, Andrew Gelman, Christian Robert, and Jennifer~L
  Tackett.
\newblock Abandon statistical significance.
\newblock {\em The American Statistician}, 73(sup1):235--245, 2019.

\bibitem{meehl1978theoretical}
Paul~E Meehl.
\newblock Theoretical risks and tabular asterisks: Sir karl, sir ronald, and
  the slow progress of soft psychology.
\newblock {\em Journal of Consulting and Clinical Psychology}, 46:806--834,
  1978.

\bibitem{rosenbaum1961moments}
S~Rosenbaum.
\newblock Moments of a truncated bivariate normal distribution.
\newblock {\em Journal of the Royal Statistical Society. Series B
  (Methodological)}, pages 405--408, 1961.

\bibitem{rozeboom1960fallacy}
William~W Rozeboom.
\newblock The fallacy of the null-hypothesis significance test.
\newblock {\em Psychological bulletin}, 57(5):416, 1960.

\bibitem{van2019default}
Erik~Willem van Zwet.
\newblock A default prior for regression coefficients.
\newblock {\em Statistical methods in medical research}, 28(12):3799--3807,
  2019.

\bibitem{wasserstein2016asa}
Ronald~L Wasserstein and Nicole~A Lazar.
\newblock The asa's statement on p-values: context, process, and purpose, 2016.

\bibitem{wasserstein2019moving}
Ronald~L Wasserstein, Allen~L Schirm, and Nicole~A Lazar.
\newblock Moving to a world beyond ``p<0.05'', 2019.

\bibitem{xu2011bayesian}
Lizhen Xu, Radu~V Craiu, and Lei Sun.
\newblock Bayesian methods to overcome the winner's curse in genetic studies.
\newblock {\em The Annals of Applied Statistics}, pages 201--231, 2011.

\end{thebibliography}

\end{document}